\documentclass[journal,onecolumn,11pt]{IEEEtran}
\usepackage[numbers,sort&compress]{natbib}
\usepackage{pslatex}
\usepackage{amsfonts,color,morefloats}
\usepackage{amssymb,amsmath,latexsym,amsthm}
\usepackage{hyperref}
\hypersetup{hidelinks}

\newcommand{\tr}{{\mathrm{Tr}}}

\newcommand{\gf}{{\mathbb{F}}}

\newcommand{\s}{{\mathbb{S}}}

\newtheorem{theorem}{Theorem}
\newtheorem{lemma}[theorem]{Lemma}
\newtheorem{proposition}[theorem]{Proposition}
\newtheorem{corollary}[theorem]{Corollary}

\newtheorem{definition}{Definition}
\newtheorem{example}{Example}

\newtheorem{remark}{Remark}

\begin{document}
\title{On the differential spectrum of a class of APN power functions over odd characteristic finite fields and their $c$-differential properties}
	\author{\IEEEauthorblockN{Haode Yan\IEEEauthorrefmark{1}, Sihem Mesnager\IEEEauthorrefmark{2}~\IEEEmembership{Member,~IEEE}, and Xiantong Tan\IEEEauthorrefmark{1}}\\
		\IEEEauthorblockA{\IEEEauthorrefmark{1}School of Mathematics, Southwest Jiaotong University, Chengdu, China. }\\
		\IEEEauthorblockA{\IEEEauthorrefmark{2}Department of Mathematics, University of Paris VIII, 93526 SaintDenis, with University Sorbonne Paris Cit{\'{e}}, LAGA, UMR 7539, CNRS, 93430 Villetaneuse, and also with the T{\'{e}}l{\'{e}}com Paris, France. }\\
		\IEEEauthorblockA{\href{mailto: hdyan@swjtu.edu.cn}{hdyan}@swjtu.edu.cn, 
		\href{mailto: smesnager@univ-paris8.fr}{smesnager}@univ-paris8.fr, \href{mailto: Xiantong@my.swjtu.edu.cn}{Xiantong}@my.swjtu.edu.cn,}\\
		\IEEEauthorblockA{Corresponding Author: Sihem Mesnager \quad Email: smesnager@univ-paris8.fr}}
\maketitle
\begin{abstract}
The differential spectrum of a cryptographic function is of important interest for estimating the resistance of the involved vectorial function to some variances of differential cryptanalysis. However, it is difficult to completely determine a power function's differential spectrum. The research in this direction is active and goes beyond the Boolean case. Only three classes of Almost Perfect Nonlinear (for short, APN) power functions over odd characteristic finite fields have been investigated in the literature, and their differential spectra were determined. In the present article, we concentrate in studying the differential and the $c$-differential uniformity (for some $c\in \gf_{p^n}\setminus\{0,1\}$) and their related differential spectrum (resp. $c$-differential spectrum) of the power functions $F(x)=x^{d}$ over the finite field $\gf_{p^n}$ of order $p^n$ (where $p$ is an odd prime) for $d=\frac{p^{n}-3}{2}$. We emphasize that, by focusing on the power functions $x^d$ with even $d$ over $\gf_{p^n}$ ($p$ odd), the power functions $F$ we consider are APN which are of the lowest differential uniformity and the nontrivial differential spectrum.
The differential uniformity of the power function $F(x)=x^{\frac{p^{n}-3}{2}}$ over the finite field $\gf_{p^n}$ of order $p^n$ (where $p$ is an odd prime), was studied by Helleseth and Sandberg in 1997, where $p^n\equiv3\pmod{4}$ is an odd prime power with $p^n>7$. It was shown that $F$ is PN when $p^n=27$, APN when $5$ is a nonsquare in $\gf_{p^n}$, and differentially $3$-uniform when $5$ is a square in $\gf_{p^n}$. In this paper, by investigating some equation systems and certain character sums over $\gf_{p^n}$, the differential spectrum of $F$ is completely determined. Moreover, we examine the extension of the so-called $c$-differential uniformity by investigating the $c$-differential properties of $F$. Specifically, an upper bound of the $c$-differential uniformity of $F$ is given, and its $c$-differential spectrum is considered in the case where $c=-1$. Despite vast recent literature since the introduction of this notion, we emphasize that compelling cryptographic motivation directly related to attacks on conventional symmetric cryptosystems was not performed yet (but its impact is justified in some instances of specific ciphers for $p=2$). Nevertheless, it remains theoretically defensible since the concept extends the differential uniformity. Finally, we emphasize that, throughout our study of the differential spectrum of the considered power functions, we provide methods for evaluating sums of specific characters with connections to elliptic curves and for determining the number of solutions of specific systems of equations over finite fields. \end{abstract}


\section{Introduction}

Substitution boxes (S-boxes for short) play a crucial role in the field of symmetric block ciphers. Let $ {\mathbb{F}_{q}} $ be the finite field with $ q $ elements, where $q$ is a prime power. For a cryptographic function $F$ from $\gf_{q}$ to itself, the main tools to handle $F$ regarding the differential attack 
introduced by Biham and Shamir \cite{BS} 
are the difference distribution table (DDT for short) and the differential uniformity introduced by Nyberg \cite{Nyberg93} in 1994. For any $ a, b \in \gf_{q}$, the DDT entry at point $ (a,b) $, denoted by $ \delta_F(a,b) $, is defined as \[\delta_F(a,b)=\big{|} \{x \in \gf_{q}: ~F(x+a)-F(x)=b\} \big{| },\]
where $ \big{|} S \big{|} $ denotes the cardinality of the finite set $S$.
The differential uniformity of the function $ F $, denoted by $ {\Delta _F} $, is defined as
\[{\Delta _F}=\max\{\delta_F(a,b): ~a \in {\mathbb{F}_{q}^*}, b \in \mathbb{F}_{q}\} ,\]
where $\gf_{q}^*=\gf_{q}\setminus\{0\}$. In the Boolean case (that is $p=2$), $F$ can be used as an S-box inside a symmetric cryptosystem (namely, stream and block ciphers); the smaller the value $\Delta_F$ is, the better the contribution of $ F $ to the resistance against differential attack. When $\Delta_{F}=1$, $F$ is said to be \emph{perfect nonlinear} (PN for short) function. 
Whereas, when $\Delta_F=2$, $F$ is called an \emph{almost perfect nonlinear} (APN for short) function. Note that PN functions over even characteristic finite fields do not exist. PN and APN functions play an important role in theory and application (up to now, the applicative aspects in cryptography concern only the Boolean case). A survey on the differential uniformity of vectorial Boolean functions can be found in the chapter of Carlet~\cite{CH1} and the recent book \cite{CarletBook}. Recent progress on cryptographic functions with low differential uniformity can be found in \cite{BD,BCHLS,BCL,CM,DO,DOB1,DOB2,HRS,HS,PT,QTLG,QTTL,TCT,TZ,ZHS,ZKW,ZW} and the references therein.

Power functions with low differential uniformity serve as good candidates for the
design of S-boxes not only because of their strong resistance to differential attacks but also for the usually low implementation cost in hardware. When $F$ is a power function, i.e., $F(x)=x^d$ for an integer $d$, one easily sees that $\delta_F(a,b)=\delta_F(1,{b/{a^d}})$ for all $a\in \gf_{q}^*$ and $b\in \gf_{q}$.
That is to say, the differential properties of $F$ are wholly determined by the values of $\delta_F(1,b)$ as $b$ runs through $\gf_{q}$. Their resistance to the standard differential attack attracted attention and was investigated. The notion of the differential spectrum of a power function was, in fact, firstly proposed by Blondeau, Canteaut, and Charpin in \cite{BCC} as follows.
\begin{definition}\label{def1}
	Let $F(x)=x^d$ be a power function over $\gf_{q}$ with differential uniformity $ \Delta_F $. Denote
	\[\omega_i= \big{|} \left\{b\in \gf_{q}: \delta_F(1, b)=i\right\} \big{|} ,\,\,0\leq i\leq \Delta_F.\]
	The differential spectrum of $F$ is defined to be the multi-set
	\[ DS_F =\left\{\omega_i>0:0\leq i \leq \Delta_F \right\}.\]
\end{definition} 
It can be seen that the value distribution of DDT of a power function can be deduced via its differential spectrum. Moreover, it has been shown in  \cite{BCC}  that the elements in the differential spectrum of $F$ satisfy the two following helpful identities in the set $\mathbb{N}$ of the positive integers.

\begin{equation}\label{omegaiomega}
	\sum_{i=0}^{\Delta _{F}}\omega _{i}=\sum_{i=0}^{\Delta _{F}}(i\cdot \omega _{i})=q.
\end{equation}
The following lemma plays an important role
in determining the differential spectrum of $F$.

\begin{lemma}\label{identicalequation}
	(\cite{HRS}, Theorem 10)
	With the notation introduced in Definition \ref{def1}, let $N_{4}$ denote the
	number of solutions $(x_{1},x_{2},x_{3},x_{4})\in (\gf_{q})^{4}$ of the equation system
	\begin{equation*}
		\Bigg\{ \begin{array}{ll}
			{x_1} - {x_2} + {x_3} - {x_4} &= 0\\
			x_1^d - x_2^d + x_3^d - x_4^d &= 0.
		\end{array}
	\end{equation*}
	Then we have 
	\begin{equation}\label{i^2omega}
		\sum_{i=0}^{\Delta _{F}}i^{2}\omega_{i}=\frac{N_{4}-q^2}{q-1}.
	\end{equation}
\end{lemma}

However, it is difficult to determine a power function's differential spectrum ultimately. The differential spectra are known only for a few classes of power functions over finite fields of odd characteristics. We summarize them in Table \ref{t1}. For $p=2$, we refer to \cite[Chapter 11]{CarletBook}. Among the known results, there are only three classes of APN functions. More precisely, the three classes of APN functions are listed below.
\begin{enumerate}
	\item[(1)] $x^{\frac{3^{n+1}-1}{4}}$ over  $\gf_{3^n}$, where $n$ is odd, which is reported by \cite{CHNC},
	\item[(2)] $x^{\frac{5^k+1}{2}}$ over $\gf_{5^n}$, where $\gcd(n,k)=1$, which is reported by \cite{CHNC}, and
	\item[(3)] $x^{3^n-3}$ over $\gf_{3^n}$, where $n$ is odd, which is reported by \cite{XZLH}.
\end{enumerate}

\begin{table}[t]
	\renewcommand{\arraystretch}{1.7}
	\caption{Power functions $F(x)=x^{d}$ over $\gf_{p^{n}}$ ($p$ is odd) with known differential spectrum}
	\label{t1}
	\centering
		\begin{tabular}{|c|c|c|c|}
			\hline
			$d$ & Conditions & $\Delta_{F}$ & Ref. \\ 
			\hline\hline
			$2\cdot 3^{\frac{n-1}2}+1$ & $n$ is odd & 4 & \cite{Dobbertin2001}\\
			\hline
			$\frac{p^k+1}2$ & $\gcd(n,k)=e$ & $\frac{p^e-1}2$ or $p^e+1$ & \cite{CHNC}\\
			\hline
			$\frac{p^n+1}{p^m+1}+\frac{p^n-1}{2}$ & $p\equiv3\pmod4$, $m\mid n$ and $n$ odd & $\frac{p^m+1}2$  & \cite{CHNC}\\
			\hline
			$p^{2k}-p^k+1$ & $\gcd(n,k)=e$,\ $\frac ne$ is odd & $p^e+1$ & \cite{Lei2021,Yan}\\
			\hline
			$p^n-3$ & any $n$ & $\leq$5 & \cite{XZLH,YXLHXL}\\
			\hline
			$p^m+2$ &  $p>3$, $n=2m$ & 4  & \cite{MXLH}\\
			\hline
			$\frac{p^n+3}2$ & $p\geq 5$, $p^n\equiv1\pmod4$ & 3 & \cite{JLLQ}\\\hline
			$\frac{5^n-3}2$ & any $n$ & 4 or 5 & \cite{YL}\\\hline
			$\frac{p^n-3}2$ & $p^n\equiv3\pmod4$, $p^{n}>7$ and $p^n\ne27$ & 2 or 3 & This paper\\\hline
	\end{tabular}
\end{table}

Let $F(x)=x^{d}$ be the power function over $\gf_{p^{n}}$, where $p$ is a prime, $p^n\equiv 3\pmod4$, $p^{n}>7$ and $d=\frac{p^{n}-3}{2}$. In \cite{HS}, Helleseth and Sandberg proved that $F$ is PN when $p^n=27$, is APN when $\chi(5)=-1$, and is differentially $3$-uniform when $\chi(5)=1$. Herein and hereafter, $\chi$ denotes the quadratic multiplicative character over $\gf_{p^n}$. This paper mainly studies the differential properties of $F$ over $\gf_{p^n}$. The differential spectrum of $F$ is determined, {which is} the fourth class of APN function with a known differential spectrum. We also studied the $c$-differential properties of $F$. The rest of this paper is organized as follows. Section \ref{A} introduces some basic notions about quadratic multiplicative character and results on quadratic multiplicative character sums,
which will play a crucial role in deriving our main results. In Section \ref{B}, we determine the number of solutions of a certain system of equations over finite fields. The differential spectrum of $F$ is investigated in Section \ref{C}, and some examples are given. In Section \ref{D}, we study the $c$-differential uniformity of $F$ and determine its $(-1)$-differential spectrum. Section \ref{E} concludes this paper.

\section{On quadratic character sums}\label{A}
 Let $\gf_{p^n}$ be the finite field with $p^n$ elements, where $p$ is an odd prime, and $n$ is a positive integer. And let $\gf_{p^n}^{*}=\gf_{p^n}\setminus\{0\}$. 
 This section mainly introduces some basic results on quadratic multiplicative character sums $\chi$ over $\gf_{p^n}$, i.e.,
\begin{equation*}
	\chi(x)=x^{\frac{p^{n}-1}{2}}=\left\{\begin{array}{ll}
		1, & \hbox{if $x$ is a square,}\\ 
		0, & \hbox{if $x=0$;}\\
		-1, & \hbox{if $x$ is a nonsquare.}
	\end{array}\right.
\end{equation*}
These exponential sums will appear naturally in our study of the differential spectra of the function $F(x)=x^{\frac{p^{n}-3}{2}}$ over the finite field $\gf_{p^n}$. These connections are highlighted below with some results in the case  $p^n \equiv 3 (\mathrm{mod}~4)$.\\

Let $\gf_{{p^{n}}}[x]$ be the polynomial ring over $\gf_{{p^{n}}}$. We consider the sum involving the quadratic multiplicative sums of the form
\begin{equation*}
	\sum_{x\in \gf_{{p^{n}}}}\chi(f(x))
\end{equation*}
with $f(x)\in \gf_{{p^{n}}}[x]$. The case of $\deg(f(x))=1$ is trivial, and for $\deg(f(x))=2$, the following explicit formula was established in \cite{FF}.

\begin{lemma}\label{2}
	(\cite{FF}, Theorem 5.48)
	Let $f(x)=a_{2}x^{2}+a_{1}x+a_{0}\in \gf_{{p^{n}}}[x]$ with $p$ odd and $a_{2} \neq 0$. Put $\Delta=a_{1}^{2}-4a_{0}a_{2}$ and let $\chi$ be the quadratic character of $\gf_{{p^{n}}}$. Then
	\begin{equation*}
		\sum_{x\in \gf_{{p^{n}}}}\chi(f(x))=\left\{\begin{array}{ll}
			-\chi(a_{2}), & \hbox{if $\Delta\neq0$,}\\ 
			(p^{n}-1)\chi(a_{2}),  & \hbox{if $\Delta=0$.}
		\end{array}\right.
	\end{equation*}
\end{lemma}
For $\deg(f(x))\geq 3$, it is challenging to derive an explicit a general formula for the character sum $\sum\limits_{x\in \gf_{p^{n}}}\chi(f(x))$. However, when $\deg(f(x))=3$, such a sum can be computed by considering rational points of elliptic curves. More specifically, for a cubic function, we denote $\Gamma _{p,n}$ as
\begin{equation}
	\label{Gamma}
	\Gamma _{p,n}=\sum_{x\in \gf_{p^{n}}}\chi(f(x)).
\end{equation}
To evaluate  $\Gamma _{p,n}$, we shall use some elementary concepts from the theory of elliptic curves. Most of the terminologies and notation we use come from \cite{SilveEC}. Let $E/\gf_{p}$ be the elliptic curve over $\gf_{p}$:
\begin{equation*}
	E:y^{2}=f(x).
\end{equation*}
Let $N_{p,n}$ denote the number of $\gf_{p^{n}}$-rational points (remember the extra point at infinity) on the curve $E/\gf_{p}$. From Subsection 1.3 in [\cite{SilveEC}, P. 139, Chap. V] and Theorem 2.3.1 in [\cite{SilveEC}, P. 142, Chap. V], $N_{p,n}$  can be computed from $\Gamma _{p,n}$. More precisely,  for every $n\geq1$,
\begin{equation*}
	N_{p,n}=p^n+1+\Gamma _{p,n}.
\end{equation*}
Moreover, 
\begin{equation}\label{gamma}
	\Gamma _{p,n}=-\alpha^{n}-\beta^{n},
\end{equation}
where $\alpha$ and $\beta$ are the complex solutions of the quadratic equation $T^{2}+\Gamma_{p,1}T+p=0$.

We  are now interested in two specific character sums $\lambda^{(1)}_{p,n}$ and $\lambda^{(2)}_{p,n}$. Let
\begin{align*}
	\lambda^{(1)}_{p,n}&=\sum_{x\in \gf_{p^{n}}}\chi(x(x^2+4x-1))
\end{align*}
and
\begin{align}\label{sum2}
	\lambda^{(2)}_{p,n}&=\sum_{x\in \gf_{p^{n}}}\chi(x(x^2-2x+5)).
\end{align}

 The former sums will be helpful in Section \ref{B}  since, as we shall see, the computation of the differential spectrum of the power function $x^{\frac{p^{n}-3}{2}}$ over $\gf_{{p^{n}}}$ boils down to evaluating those character sums.

In the following examples, we give the exact values of $\lambda^{(1)}_{p,n}$ and $\lambda^{(2)}_{p,n}$, 
respectively,  over prime fields (for specific values of $p$).
\begin{example}
\begin{itemize}
\item Let $p=7$. For $n=1$, one has $\lambda^{(1)}_{7,1}=2$. The quadratic equation $T^{2}+2T+7=0$ has two complex roots $-1\pm \sqrt{-6}$. By \eqref{gamma},  we have
	\begin{align*}
		\lambda^{(1)}_{7,n}&=-(-1+\sqrt{-6})^{n}-(-1-\sqrt{-6})^{n}\\
		&=(-1)^n\cdot2\sum\limits_{k = 0}^{\left\lfloor {\frac{n}{2}} \right\rfloor } {{{\left( { - 1} \right)}^{k+1}} \binom{n}{2k} {6^{k}}}.
	\end{align*}

\item
	Let $p=11$. For $n=1$, one has  $\lambda^{(2)}_{11,1}=0$. The quadratic equation $T^{2}+11=0$ has two complex roots $\pm \sqrt{-11}$. Similarly, we have 
	\begin{align*}
		\lambda^{(2)}_{11,n}&=-(\sqrt{-11})^{n}-(-\sqrt{-11})^{n},\\
		&=\left\{\begin{array}{ll}
			(-1)^{\frac{n}{2}+1}\cdot 2\cdot {11}^{\frac{n}{2}}, & \hbox{ $n$ is even,}\\
			0, & \hbox{ $n$ is odd.}
		\end{array}\right.
	\end{align*} 
	\end{itemize}
\end{example}
In addition, we have the following bound on $\lambda^{(i)}_{p,n}$ for $i=1,2$.
\begin{theorem}[\cite{SilveEC}, Corollary 1.4]\label{bound}
	\label{bound}
With the notation as above, we have
	\[|\lambda^{(i)}_{p,n}|\leq 2p^{\frac{n}{2}}\]
	for $i=1,2$.
	\end{theorem}
We define 
\begin{align}\label{sum}
	\lambda_{p,n}:=\lambda^{(1)}_{p,n}+\lambda^{(2)}_{p,n}=\sum_{x\in \gf_{p^{n}}}\chi(x(x^2+4x-1))+\sum_{x\in \gf_{p^{n}}}\chi(x(x^2-2x+5)),
\end{align}
{which will be used later}.

At the end of this section, we give the following results on certain character sums. 
\begin{lemma}\label{sums}
	Let $p^n \equiv 3 (\mathrm{mod}~4)$, we have

	\begin{align*}
		&(i)\sum_{x\in\gf_{p^n}}\chi(x(x^2-1))=0,\\
		&(ii) \sum_{x\in\gf_{p^n}}\chi(x(x^2-1)(x^2+4x-1))=0, and\\
		&(iii)\sum_{x\in \gf_{{p^{n}}}}\chi((x^{2}-1)(x^{2}+4x-1))=\lambda^{(2)}_{p,n}-1,
	\end{align*}
where $\lambda^{(2)}_{p,n}$ is defined in \eqref{sum2}.
\end{lemma}
\begin{proof}
	\begin{enumerate}
		\item[(i).] We have $\sum\limits_{x\in\gf_{p^n}}\chi(x(x^2-1))=\sum\limits_{x\in\gf_{p^n}^*}\chi(x(x^2-1))$. Since $x^{-1}$ permutes $\gf_{{p^n}}^*$, let $y=x^{-1}$, then 
		\[\sum_{x\in\gf_{p^n}^*}\chi(x(x^2-1))=\sum_{y\in\gf_{p^n}^*}\chi(y^{-1}(y^{-2}-1))=\sum_{y\in\gf_{p^n}^*}\chi(y(1-y^2))=\chi(-1)\sum_{y\in\gf_{p^n}^*}\chi(y(y^2-1)).\]
		We conclude that $\sum\limits_{x\in\gf_{p^n}^*}\chi(x(x^2-1))=0$ since $\chi(-1)=-1$. Consequently, $\sum\limits_{x\in\gf_{p^n}}\chi(x(x^2-1))=0$.
		\item[(ii).] Note that $-x^{-1}$ permutes $\gf_{p^n}^*$. The proof is similar to that of (i), and we omit the details.
		\item[(iii).] It is  not difficult to notice that
		\[
		\sum\limits_{x\in \gf_{{p^{n}}}}\chi((x^{2}-1)(x^{2}+4x-1))=\sum\limits_{x\in \gf_{p^n},x\neq\pm1}\chi((x^{2}-1)(x^{2}+4x-1))=\sum\limits_{x\in \gf_{p^n},x\neq\pm1}\chi(\frac{x^{2}+4x-1}{x^{2}-1}).\]
		Let $\frac{x^{2}+4x-1}{x^{2}-1}=u$. Then $u$ and $x$ satisfy 
		\begin{equation}\label{u-xrelation}
			(u-1)x^2-4x-(u-1)=0.
		\end{equation}
		Note  that $x=0$ if and only if $u=1$. When $u\ne1$, \eqref{u-xrelation} is a quadratic equation on the variable $x$, and its discriminant is $4(u^2-2u+5)$, which is not $0$ since $-1$ is a nonsquare. Therefore, for $u\neq1$, \eqref{u-xrelation} has two distinct solutions if and only if $\chi(u^{2}-2u+5)=1$. Then we have
		\begin{equation}\label{u2}
			\sum_{x\in \gf_{{p^{n}}}, x\neq\pm1}\chi(\frac{x^{2}+4x-1}{x^{2}-1})=\chi(1)+2\sum_{u\in \mathcal{U}}\chi(u),
		\end{equation}
		where
		\begin{equation*}
			\mathcal{U}:=\{u\in \gf_{{p^{n}}}|u\ne1, \chi(u^{2}-2u+5)=1\}.
		\end{equation*}
		Furthermore,
		\begin{align*}
			&2\sum_{u\in \mathcal{U} }\chi(u)\\=&\sum_{u\ne 1}(1+\chi(u^{2}-2u+5))\chi(u)\\=&\sum_{u\in \gf_{{p^{n}}}}(1+\chi(u^{2}-2u+5))\chi(u)-2\chi(1)\\=&\sum_{u\in \gf_{{p^{n}}}}\chi(u)+\sum_{u\in \gf_{{p^{n}}}}\chi(u(u^2-2u+5))-2\chi(1)\\=&\lambda^{(2)}_{p,n}-2
		\end{align*}
		since $\sum\limits_{u\in \gf_{{p^{n}}}}\chi(u)=0$,  which completes the proof  thanks to $\eqref{u2}$.
	\end{enumerate}
\end{proof}

\section{On the number of solutions of certain equation systems}\label{B}

In this section we  are interested in  computing the  number of solutions $(y_1,y_2,y_3)\in(\gf_{p^n}^*)^3$ of the equation system
	\begin{equation}\label{n_4equationsystem}
	\left\{ \begin{array}{ll}
		{y_1} - {y_2} + {y_3} - 1 &= 0\\
		{y^d_1} - {y^d_2} + {y^d_3} - 1 &= 0,
	\end{array} \right.
\end{equation}
over $\gf_{p^n}$. Determining such a number will play an important role in computing the differential spectrum of $x^{d}$. Let $n_{(i,j,k)}$ denote the number of solutions $(y_1,y_2,y_3)$ of Equation \eqref{n_4equationsystem} with $(\chi(y_1), \chi(y_2), \chi(y_3))=(i,j,k)$, where $i,j,k\in\{\pm 1\}$. First, the following insertions in Proposition  \ref{property} come straightforwardly and  will therefore be stated without proof.

\begin{proposition}\label{property}With the notation as above, we have
	\begin{enumerate}
		\item[(i).] $(y_1,y_2,y_3)$ is a solution of Equation \eqref{n_4equationsystem} if and only if $(y_3,y_2,y_1)$ is a solution of Equation \eqref{n_4equationsystem}, then $n_{(1,1,-1)}=n_{(-1,1,1)}$ and $n_{(1,-1,-1)}=n_{(-1,-1,1)}$.
		\item[(ii).] $(y_1,y_2,y_3)$ is a solution of Equation \eqref{n_4equationsystem} if and only if $(\frac{y_1}{y_2},\frac{1}{y_2},\frac{y_3}{y_2})$ is a solution of Equation \eqref{n_4equationsystem}, then $n_{(-1,-1,-1)}=n_{(1,-1,1)}$.
		\item[(iii).] $(y_1,y_2,y_3)$ is a solution of Equation \eqref{n_4equationsystem} if and only if $(\frac{1}{y_1},\frac{y_3}{y_1},\frac{y_2}{y_1})$ is a solution of Equation \eqref{n_4equationsystem}, then $n_{(-1,1,1)}=n_{(-1,-1,-1)}$.
	\end{enumerate}
\end{proposition}

In the case when $p$ is such that $p^n \equiv 3 (\mathrm{mod}~4)$ and when $d=\frac{p^n-3}{2}>0$, we have the following result.

\begin{theorem}\label{n_4value}
	Let $p^n \equiv 3 (\mathrm{mod}~4)$ and $d=\frac{p^n-3}{2}>0$ is a positive integer. The number of solutions $(y_1,y_2,y_3)\in(\gf_{p^n}^*)^3$ of Equation \eqref{n_4equationsystem}
is $\frac{1}{2}(5p^{n}+\lambda_{p,n}-13)$, where $\lambda_{p,n}$ was defined in \eqref{sum}.
\end{theorem}
\begin{proof}
	 We denote by $n_4$ the number of solutions $(y_1,y_2,y_3)\in(\gf_{p^n}^*)^3$ of Equation \eqref{n_4equationsystem}. With the notation as above, we obtain
	 $n_{(1,1,-1)}=n_{(-1,1,1)}=n_{(-1,-1,-1)}=n_{(1,-1,1)}$ and $n_{(1,-1,-1)}=n_{(-1,-1,1)}$ by Proposition \ref{property}. Note that $y_{i}\ne 0, {y^d_i}=\chi(y_i)y^{-1}_i$, we discuss Equation \eqref{n_4equationsystem} in the following cases.
	 \begin{enumerate}
	
\item Case \rm\uppercase\expandafter{\romannumeral1}.
	 $(\chi(y_{1}), \chi(y_2), \chi(y_3))=(1,1,1)$. Then Equation  \eqref{n_4equationsystem} becomes
	\begin{equation}\label{n(1,1,1)}
		\left\{ \begin{array}{ll}
			{y_1} - {y_2} + {y_3} - 1 &= 0\\
			{y^{-1}_1} - {y^{-1}_2} + {y^{-1}_3}-1 &= 0.
		\end{array} \right.
	\end{equation}
	If $y_{3}=1$, then $y_{2}=y_{1}$, $ (y_{1},y_{1},1) $ is a solution of (\ref{n(1,1,1)}) when $\chi(y_{1})=1$. We obtain $\frac{p^{n}-1}{2}$ solutions. If $y_{3}\ne 1$,  then $ y_{1}\neq y_{2} $, we have 
	\[y_{1}y_{2}=-(y_{1}-y_{2})({y^{-1}_{1}}-{y^{-1}_{2}})^{-1}=-(y_{3}-1)({y^{-1}_3}-1)^{-1}=y_{3}.\]
	Then we obtain $(y_{1}-1)(y_{2}+1)=0$ from the first equation of (\ref{n(1,1,1)}), then $y_{1}=1$ since $\chi(-1)=-1$ and $\chi(y_2)=1$. We conclude that the equation system (\ref{n(1,1,1)}) has solutions with type $(1,y_{3},y_{3})$ when $\chi(y_{3})=1$ and $y_3\neq 1$. By the above discussion, we obtain $n_{(1,1,1)} =2\cdot\frac{p^{n}-1}{2}-1=p^{n}-2$.

	\item Case \rm\uppercase\expandafter{\romannumeral2}.
	 $(\chi(y_{1}), \chi(y_2), \chi(y_3))=(1,1,-1)$. Then Equation \eqref{n_4equationsystem} becomes 
\begin{equation}\label{n(1,1,-1)}
	\left\{ \begin{array}{ll}
		{y_1} - {y_2} + {y_3} - 1 &= 0\\
		{{y^{-1}_1}}- {{y^{-1}_2}} - {{y^{-1}_3}} - 1 &= 0.
	\end{array} \right.
\end{equation}
Note that (\ref{n(1,1,-1)}) has no solution when $y_{3}=\pm 1 $, we obtain $ y_{1}-y_{2}=-y_3+1\ne 0 $, and \[-y_{1}y_{2}=(y_{1}-y_{2})({{y^{-1}_{1}}-{y^{-1}_{2}}})^{-1}=-(y_3-1)(y^{-1}_3+1)^{-1} =-\frac{y_{3}(y_{3}-1)}{y_{3}+1}.\]
We know that $y_1$ and $-y_2$ satisfy the following quadratic equation on the variable $t$.
\begin{equation}\label{n(1,1,-1)simple}
	t^2+(y_{3}-1)t-\frac{y_{3}(y_{3}-1)}{y_{3}+1}=0.
\end{equation}
The {discriminant} of (\ref{n(1,1,-1)simple}) is $\Delta=\frac{(y_{3}-1)(y_{3}^{2}+4y_{3}-1)}{y_{3}+1}$. On one hand, if \eqref{n(1,1,-1)} has a solution $(y_1,y_2,y_3)$ with $(\chi(y_1),\chi(y_2),\chi(y_3))=(1,1,-1)$, then $\chi(\Delta)=1$ ($y_1\neq -y_2$) and $\chi(-y_1y_2)=\chi(-\frac{y_{3}(y_{3}-1)}{y_{3}+1})=-1$. On the other hand, if $\chi(\Delta)=1$ and $\chi(-\frac{y_{3}(y_{3}-1)}{y_{3}+1})=-1$ for some $y_3$ with $\chi(y_3)=-1$, then \eqref{n(1,1,-1)simple} has two distinct solutions and their product is a nonsquare. We assert that one of the two solutions is a square and the other is a nonsquare; the square solution is $y_1$ while the nonsquare solution is $-y_2$. Then we obtain a solution $(y_1,y_2,y_3)$ of \eqref{n(1,1,-1)} with $(\chi(y_1),\chi(y_2),\chi(y_3))=(1,1,-1)$.
We conclude that
\begin{align*}
	n_{(1,1,-1)}&=\#\{y_3\in\gf_{p^n}, y_3\neq 0, \pm1 : \chi(y_3)=-1, \chi(\frac{(y_{3}-1)(y_{3}^{2}+4y_{3}-1)}{y_{3}+1})=1, \chi(-\frac{y_{3}(y_{3}-1)}{y_{3}+1})=-1\}
	\\&=\#\{y_3\in\gf_{p^n}, y_3\neq 0, \pm1 : \chi(y_3)=-1, \chi(y_{3}^{2}+4y_{3}-1)=-1, \chi(y^2_3-1)=-1\}.
\end{align*}
Then
\begin{align*}
	n_{(1,1,-1)}=&\frac{1}{8}\sum_{\substack{y_{3}\neq 0,\pm 1,\\y_{3}^{2}+4y_{3}-1\neq0}}(1-\chi(y_{3}))(1-\chi(y_{3}^{2}-1))(1-\chi(y_{3}^{2}+4y_{3}-1))\notag\\={}&\frac{1}{8}\sum_{y_3\in\gf_{p^n}}(1-\chi(y_{3}))(1-\chi(y_{3}^{2}-1))(1-\chi(y_{3}^{2}+4y_{3}-1))\notag\\{}&-\frac{1}{8}((1-\chi(-1))^2+0+(1-\chi(-1))(1-\chi(-4))+0)\notag\\={}&-1+\frac{1}{8}\sum_{y_3\in\gf_{p^n}}(1-\chi(y_{3}))(1-\chi(y_{3}^{2}-1))(1-\chi(y_{3}^{2}+4y_{3}-1)).
\end{align*}
The above identities hold since when $y^2_3+4y_3-1=0$, then $y_3\neq 0$ and $(1-\chi(y_3))(1-\chi(y^2_3-1))=(1-\chi(y_3))(1-\chi(-4y_3))=1-(\chi(y_3))^2=0$.
By Lemmas \ref{2} and \ref{sums}, we have
\begin{align*}
	n_{(1,1,-1)}=	&-1+\frac{1}{8}\sum_{y_3\in\gf_{p^n}}(1-\chi(y_{3}))(1-\chi(y_{3}^{2}-1))(1-\chi(y_{3}^{2}+4y_{3}-1))\notag\\={}&-1+\frac{1}{8}\sum_{y_{3}\in\gf_{p^n}}[1-\chi(y_3)-\chi(y_{3}^{2}-1)-\chi(y_{3}^{2}+4y_{3}-1)+\chi(y_3(y^2_3-1))\notag\\{}&+\chi(y_3(y^2_3+4y_3-1))+\chi((y^2_3-1)(y^2_3+4y_3-1))-\chi(y_3(y^2_3-1)(y^2_3+4y_3-1))]
	\notag\\={}&\frac{1}{8}(p^n+\lambda_{p,n}-7),
\end{align*}
where $\lambda_{p,n}$ was defined in \eqref{sum}.

\item Case \rm\uppercase\expandafter{\romannumeral3}.
 $(\chi(y_{1}), \chi(y_2), \chi(y_3))=(1,-1,-1)$. Then Equation \eqref{n_4equationsystem} becomes
	\begin{equation}\label{n(1,-1,-1)}
		\left\{ \begin{array}{ll}
			{y_1} - {y_2} + {y_3} - 1 &= 0\\
			{{y^{-1}_1}} + {{y^{-1}_2}} - {{y^{-1}_3}} - 1 &= 0.
		\end{array} \right.
	\end{equation}
It is easy to see that when $y_1=1$, (\ref{n(1,-1,-1)}) has solutions with type $(1,y_2,y_2)$, where $\chi(y_2)=-1$. If $y_1\neq 1$,  we obtain $y_{2}-y_{3}=y_1-1\neq 0$ and then \[y_{2}y_{3}=-(y_{2}-y_{3})(y_{2}^{-1}-y_{3}^{-1})^{-1}=-(y_{1}-1)(-y_{1}^{-1}+1)^{-1}=-y_{1},\]
which is impossible since $\chi(y_1)=1$ and $\chi(y_2)=\chi(y_3)=-1$. We assert that (\ref{n(1,-1,-1)}) has no solution when $y_1\neq 1$. Hence $n_{(1,-1,-1)}=\frac{p^n-1}{2}$.

	\item Case \rm\uppercase\expandafter{\romannumeral4}. $(\chi(y_{1}), \chi(y_2), \chi(y_3))=(-1,1,-1)$. Then Equation \eqref{n_4equationsystem} becomes
	\begin{equation}\label{n(-1,1,-1)}
		\left\{ \begin{array}{ll}
			{y_1} - {y_2} + {y_3} - 1 &= 0\\
			-{y^{-1}_1} - {y^{-1}_2} - {y^{-1}_3} - 1 &= 0.
		\end{array} \right.
	\end{equation}
	 Note that (\ref{n(-1,1,-1)}) has no solution when $y_{2}=-1$. If $y_2\neq -1$, then we have $y_1+y_3=y_2+1\neq 0$ and 
	 \[y_1y_3=(y_1+y_3)(y^{-1}_1+y^{-1}_3)^{-1}=(y_2+1)(-y^{-1}_2-1)^{-1}=-y_2,\]
	 which is impossible since $\chi(y_2)=1$ and $\chi(y_1)=\chi(y_3)=-1$. We assert that (\ref{n(-1,1,-1)}) has no solution. Hence $n_{(-1,1,-1)}=0$.

By discussions as above, we have \[n_{4}=n_{(1,1,1)}+4n_{(1,1,-1)}+2n_{(1,-1,-1)}+n_{(-1,1,-1)}=\frac{1}{2}(5p^{n}+\lambda_{p,n}-13),\] 
which completes the proof.
\end{enumerate}
\end{proof}

\section{The differential spectrum of the power function $x^{\frac{p^n-3}{2}}$ over $\gf_{p^n}$}\label{C}

First of all, we start with a meaningful discussion. For a power function $F(x)=x^d$ over $\gf_{p^n}$, if $\Delta_F=1$, this only occurs when $p$ is odd. The differential spectrum of $F$ is trivial, {which} equals $\{\omega_1=p^n\}$. If $\Delta_F=2$ and $p=2$, the differential spectrum of $F$ is also trivial, {which equals} $\{\omega_0=\omega_2=2^{n-1}\}$. If $\Delta_F=2$, $p$ is odd, and $d$ is odd, note that the solutions of the equation $(x+1)^d-x^d=b$ come in pairs, i.e., $x$ is a solution if and only if $-x-1$ is also a solution. Then only for $x=-\frac{1}{2}$, the corresponding $b$ has odd solution (one solution), i.e., $\omega_1=1$. Then, in this case, the differential spectrum of $F$ is $\{\omega_0=\frac{p^n-1}{2}, \omega_1=1,\omega_2=\frac{p^n-1}{2}\}$, which is also trivial.
In this section, we shall concentrate on studying the differential spectrum of the power function $F(x)=x^{d}$ over $\gf_{p^n}$, where $d=\frac{p^n-3}{2}$, $p^n\equiv3\pmod4$, $p^n>7$ and $p^n\ne27$. From the above discussion, the power functions we considered are APN power functions $x^d$ with even $d$ over $\gf_{p^n}$ ($p$ odd), which are the power functions with the lowest differential uniformity and nontrivial differential spectrum.

Now, note that  \cite{HS}, Helleseth and Sandberg proved that $F$ is APN when $\chi(5)=-1$ and is differentially $3$-uniform when $\chi(5)=1$. Denote by \[\delta(b)=|\{x\in\gf_{p^n}:(x+1)^d-x^d=b\}|\]
for any $b\in\gf_{{p^{n}}}$. It was proved in \cite{HS} that $\delta(b)\leq 2$ for $b\in\gf_{p^n}\setminus\{\pm1\}$. The following lemma on the values of $\delta(1)$ and $\delta(-1)$ was also given in the proof of Theorem 2, \cite{HS}. For the convenience,
we give an alternative proof.

\begin{lemma}\label{delta1}
	With the notation as above, we have $\delta(1)=\delta(-1)=1$ when $\chi(5)=-1$ and $\delta(1)=\delta(-1)=3$ when $\chi(5)=1$.
\end{lemma}

\begin{proof}
Obviously, $x_0$ is a solution of 
	\begin{equation}\label{delta1eqn}
		(x+1)^d-x^d=1
	\end{equation}
	if and only if $-x_0-1$ is a solution of $(x+1)^d-x^d=-1$. Then $\delta(-1)=\delta(1)$. We only consider Equation  (\ref{delta1eqn}).  It is easy to see that $x=0$ is a solution of Equation (\ref{delta1eqn})  and $x=-1$ is not a solution of Equation (\ref{delta1eqn}). For $x\neq0,-1$, $x^d=\chi(x)x^{-1}$, we distinguish the following four cases.	
	\begin{enumerate}
\item	 Case 1. $(\chi(x+1),\chi(x))=(1,1)$. Equation (\ref{delta1eqn}) becomes $x(x+1)=-1$, then $\chi(x(x+1))=\chi(-1)=-1$, which is a contradiction. Hence Equation (\ref{delta1eqn})  has no solution in this case.

\item Case 2. $(\chi(x+1),\chi(x))=(1,-1)$. Equation  (\ref{delta1eqn}) becomes $x^2-x-1=0$, the discriminant of this quadratic equation is $5$. If $\chi(5)=-1$, then Equation (\ref{delta1eqn})  has no solution in this case. If $\chi(5)=1$, $x^2-x-1=0$ has two distinct solutions, namely, $x_1$ and $x_2$. Without loss of generality, we assume that $\chi(x_1)=-1$ and $\chi(x_2)=1$ since $\chi(x_1x_2)=\chi(-1)=-1$. Moreover, $x_1$ satisfies $x_1+1=x_1^2$, which means $x_1+1$ is a square element. Then $x_1$ is the unique solution of  Equation (\ref{delta1eqn})  in this case when $\chi(5)=1$.
\item Case 3. $(\chi(x+1),\chi(x))=(-1,1)$. Equation (\ref{delta1eqn}) becomes $x^2+3x+1=0$, then $x=-(x+1)^2$, which contradicts to $\chi(x)=1$. Hence, Equation (\ref{delta1eqn})  has no solution in this case.

\item Case 4. $(\chi(x+1),\chi(x))=(-1,-1)$. Equation (\ref{delta1eqn})  becomes $x^2+x-1=0$, the discriminant of this quadratic equation is also $5$. By a similar proof to Case 2, we can obtain that Equation (\ref{delta1eqn})  has no solution in this case when $\chi(5)=-1$ and has a unique solution in this case when $\chi(5)=1$.
\end{enumerate}

By the discussions above, the desired result follows.
\end{proof}

One can immediately deduce the following corollary.
\begin{corollary}\label{omega3}
	We have, $\omega_3=2$ when $\chi(5)=1$, and $\omega_3=0$ when $\chi(5)=-1$.
\end{corollary}

To determine the differential spectrum of $F$, it remains to calculate the number of solutions of the equation system given in Lemma \ref{identicalequation}. Recall that $N_{4}$ denotes the number of solutions $(x_{1},x_{2},x_{3},x_{4})\in (\gf_{p^{n}})^{4}$ of the equation system
\begin{equation}\label{equationsystem}
	\left\{ \begin{array}{ll}
		{x_1} - {x_2} + {x_3} - {x_4} &= 0\\
		{x_1^d} - {x_2^d} + {x_3^d} - {x_4^d} &= 0
	\end{array} \right.
\end{equation}
In the following theorem, we determine the value of $N_4$.
\begin{theorem}\label{N_4value}
	Let $p^n\equiv3\pmod4$ and $d=\frac{p^n-3}{2}>0$ is a positive integer, we have \[{{}N_4} = \left\{ \begin{array}{ll}
		\frac{1}{2}(5p^{2n}+(\lambda_{p,n}-10)p^n-(\lambda_{p,n}-7)),&~\mathrm{if}~\chi(5)=-1,\\
		\frac{1}{2}(5p^{2n}+(\lambda_{p,n}+6)p^n-(\lambda_{p,n}+9)),&~\mathrm{if}~\chi(5)=1,
	\end{array} \right.\]  
where $\lambda_{p,n}$ was defined in \eqref{sum}.
\end{theorem}
\begin{proof}
	For a solution $ (x_1,x_2,x_3,x_4) \in (\gf_{p^{n}})^{4}$ of (\ref{equationsystem}), first we consider that there exists $ x_i=0 $ for some $ 0 \le i \le 4 $. It is easy to see that $ (0,0,0,0) $ is a solution of (\ref{equationsystem}), and (\ref{equationsystem}) has no solution containing only three zeros. If there are only two zeros in  $ (x_1,x_2,x_3,x_4) $, one can get that $(0,0,x_0,x_0)$, $(x_0,0,0,x_0)$, $(x_0,x_0,0,0)$ and $(0,x_0,x_0,0)$ are solutions of (\ref{equationsystem}), where $x_0\in\gf_{{p^n}}^*$. That means, (\ref{equationsystem}) has $ 4(p^n-1) $ solutions containing only two zeros. We consider that there is only one zero in $ (x_1,x_2,x_3,x_4) $. Without loss of generality, we assume that $x_4=0$, then $x_1,x_2,x_3\neq 0$ and they satisfy 
	\begin{equation}\label{equ3}
		\left\{ \begin{array}{ll}
			{x_1} - {x_2} + {x_3} &= 0\\
			{x^d_1} - {x^d_2} + {x^d_3} &= 0.
		\end{array}\right.
	\end{equation}
	Let $ y_i=\frac{x_i}{x_3} $ for $ i=1,2 $, we have $ y_1-y_2+1=0 $ and $ {y^d_1}-{y^d_2}+1=0 $ with $ y_1,y_2 \ne 0 $. Then $ y_2=y_1+1 $ and $ (y_1+1)^d-y^d_1=1 $. By the Lemma \ref{delta1}, we know that equation $ (y_1+1)^d-y^d_1=1 $ has $ \delta(1)-1 $ nonzero solutions in $ \gf_{{p^n}}^{*} $, and $y_1=-1$ is not a solution. Hence we assert \eqref{equ3} has $(\delta(1)-1)(p^n-1)$ solutions with $x_1,x_2,x_3\neq0$. Similarly, we can determine the number of solutions of (\ref{equationsystem}) with only $x_i=0$, $i=1,2$ and $3$. Then (\ref{equationsystem}) has $ 4(\delta(1)-1)(p^n-1) $ solutions containing only one zero. We conclude that (\ref{equationsystem}) has $ 1+4\delta(1)(p^n-1) $ solutions containing zeros. 

	Next we consider $ x_i \ne 0 $ for $ 1 \le i \le 4 $. Let $ y_i=\frac{x_i}{x_4} $ for $ i=1,2,3 $. We have $y_i\neq 0$ and they satisfy
	\begin{equation}\label{n4value}
		\left\{ \begin{array}{ll}
			{y_1} - {y_2} + {y_3} - 1 &= 0\\
			{y^d_1} - {y^d_2} + {y^d_3} - 1 &= 0.
		\end{array} \right.
	\end{equation}
	By Theorem \ref{n_4value}, we know that the number of solutions of \eqref{n4value} is $\frac{1}{2}(5p^{n}+\lambda_{p,n}-13)$, where $\lambda_{p,n}$ was defined in \eqref{sum}. Then $N_{4}=1+4\delta(1)(p^n-1)+n_4(p^n-1)$ follows by the value of $\delta(1)$ in Lemma \ref{delta1}. We complete the proof.
\end{proof}

We are now in a position to determine the differential spectrum of $F$, which is one of the main results of the article.
 
\begin{theorem}\label{differential spectrum}Let $p^n\equiv3\pmod4$, $p^n>7$ and $p^n\neq27$.
	The power function $F(x)=x^{\frac{p^{n}-3}{2}}$ over  $\gf_{p^{n}}$ is APN with the differential spectrum
	\begin{align*}
		DS_{F}=\{&\omega_{0}=\frac{1}{4}(p^n+\lambda_{p,n}-7),\\&\omega_{1}=\frac{1}{2}(p^n-\lambda_{p,n}+7),\\&
		\omega_{2}=\frac{1}{4}(p^n+\lambda_{p,n}-7)\}
	\end{align*}
when $\chi(5)=-1$, and is differentially 3-uniform {with the differential spectrum}
\begin{align*}
	DS_{F}=\{&\omega_{0}=\frac{1}{4}(p^n+\lambda_{p,n}+1),\\&\omega_{1}=\frac{1}{2}(p^n-\lambda_{p,n}+3),\\&
	\omega_{2}=\frac{1}{4}(p^n+\lambda_{p,n}-15),\\&
	\omega_{3}=2\}
\end{align*}
when $\chi(5)=1$, where $\lambda_{p,n}$ was defined  in \eqref{sum}.
\end{theorem}
\begin{proof}
	By \eqref{i^2omega} and Theorem \ref{N_4value},  the elements in the differential spectrum of $F$ satisfy
	\begin{equation}\label{i2omegavalue}
		\sum_{i=0}^{\Delta_{F}}i^{2}\omega_{i}=\left\{\begin{array}{ll}
			\frac{1}{2}(3p^n+\lambda_{p,n}-7), & \mathrm{if}~ \chi(5)=-1 \\ 
			\frac{1}{2}(3p^n+\lambda_{p,n}+9), &\mathrm{if}~ \chi(5)=1 .
		\end{array}\right.
	\end{equation}
By Corollary \ref{omega3}, we have $\omega_{3}=2$ when $\chi(5)=1$, and $\omega_{3}=0$ when $\chi(5)=-1$. By \eqref{omegaiomega} and \eqref{i2omegavalue}, $\omega_{0}, \omega_{1}$ and $\omega_{2}$ satisfy 
\begin{equation*}
    \left\{\begin{array}{ll}
    	\omega_{0}+\omega_{1}+\omega_{2}&=p^n\\ 
    	\omega_{1}+2\omega_{2}&=p^{n}\\
    	\omega_{1}+4\omega_{2}&=\frac{1}{2}(3p^n+\lambda_{p,n}-7)
    \end{array}\right.
\end{equation*}
when $\chi(5)=-1$, and they satisfy
\begin{equation*}
	\left\{\begin{array}{ll}
		\omega_{0}+\omega_{1}+\omega_{2}&=p^n-2\\ 
		\omega_{1}+2\omega_{2}&=p^{n}-6\\
		\omega_{1}+4\omega_{2}&=\frac{1}{2}(3p^n+\lambda_{p,n}-27)
	\end{array}\right.
\end{equation*}
when $\chi(5)=1$. By solving the above   system of equation, the differential spectrum of $F$ follows. 
\end{proof}

Below, we explicit the differential spectrum of $F(x)$ over $\gf_{p^{n}}$ for some specific values of $p$ and $n$.
\begin{example}
\begin{itemize}
	\item For $p=3,\ n=5$, the power function $F(x)=x^{120}$ over $\gf_{3^{5}}$  is APN  whose differential spectrum is
	\begin{equation*}
		DS_{F}=\{\omega_{0}=60,\omega_{1}=123,\omega_{2}=60\}.
	\end{equation*}

\item For
	 $p=7,\ n=3$, the power function $F(x)=x^{170}$ over $\gf_{7^{3}}$ is APN  whose differential spectrum is
	\begin{equation*}
		DS_{F}=\{\omega_{0}=84,\omega_{1}=175,\omega_{2}=84\}.
	\end{equation*}
\item For $p=11,\ n=3$, the power function $F(x)=x^{664}$ over $\gf_{11^{3}}$ is $3$-differentially uniform   whose differential spectrum is
	\begin{equation*}
		DS_{F}=\{\omega_{0}=316,\omega_{1}=701,\omega_{2}=312,\omega_{3}=2\}.
	\end{equation*}
	\end{itemize}
\end{example}

\section{On the $c$-differential properties of the power function $x^{\frac{p^n-3}{2}}$ over $\gf_{p^n}$}\label{D}

In \cite{BCJW}, Borisov \textit{et al.} used a new type of differential, namely multiplicative differential, that is quite useful from a practical perspective for ciphers that utilize modular multiplication as a primitive operation. Despite vast recent literature since its introduction, we emphasize that compelling cryptographic motivation directly related to attacks was not performed yet except for some specific ciphers. It seems indeed helpful to evaluate the resistance of S-boxes (in even characteristic) against cryptanalyses of some existing ciphers, such as a variant of the well-known IDEA cipher. In addition, recall that Borisov \textit{et al.} argued that one should look at other types of differential for a cryptographic function $ F $, not only the usual $ (F(x+a),F(x)) $ but $ (F(x+a),cF(x)) $. Moreover, they introduced the $c$-Differential Distribution Table ($c$DDT for short). For a function $F$ from $\gf_{{q}}$ to itself and $c\in\gf_{{q}}$, the entry at $(a,b)$ of the $c$DDT, denoted by $_c\delta_F(a,b)$, is defined as

\[_c\delta_F(a,b)=|\{x\in\gf_q: ~F(x+a)-cF(x)=b\}|.\]
The corresponding $c$-differential uniformity is defined as follows.
\begin{definition}(\cite{EFRST})
	Let $ \gf_{q} $ denote the finite field with $ q $ elements, where $ q $ is a prime power. For a function $ F:{\gf_{q}} \to {\gf_{q}} $, and $a, b,  c\in \gf_{q} $, we call 
	\[_c{\Delta _F} = \max \left\{ {_c{\delta _F}\left(a, b \right):~a,b \in {\gf_{{q}}},\mathrm{and}~a \ne 0~\mathrm{if}~c = 1} \right\}\]
	the $ c $-differential uniformity of $ F $.
\end{definition} If $ _c{\Delta _F}=\delta $, then we say that $ F $ is differentially $ (c,\delta) $-uniform. Similarly, the smaller the value $_c\Delta_F$ is, the better $F$ resists multiplicative differential attacks. If the $ c $-differential uniformity of $ F $ equals $ 1 $, then $ F $ is called a perfect $ c $-nonlinear (P$ c $N) function. P$c$N functions over odd characteristic finite fields are called $ c $-planar functions. If the $ c $-differential uniformity of $ F $ is $ 2 $, then $ F $ is called an almost perfect $ c $-nonlinear (AP$ c $N) function. It is easy to conclude that for $ c=1 $ and $ a\ne 0 $, the $ c $-differential uniformity becomes the usual differential uniformity, and the P$ c $N, and AP$ c $N functions become PN and APN functions, respectively. For cryptographic functions with low $c$-differential uniformity, the readers can refer to \cite{BC,BT,EFRST,HPRS,MRSYZ,TZJT,WZ,WLZ,YAN,ZH}. The reader could remember that the notion of $c$-differential uniformity is theoretically interesting since it is an extended version of the well-known differential uniformity. A recent survey on $c$-differential uniformity and related notions can be found in \cite{MMM-2022}. 

Similarly, when $F(x)=x^d $ is a power function, one easily observe  that $_c{\delta_F(a,b)}$=$_c{\delta_F(1,{b/{a^d}})}$ for all $a\in \gf_{{q}}^*$ and $b\in \gf_{q}$. More precisely, it was proved in \cite{MRSYZ} that the $c$-differential uniformity of the power function $F(x)=x^d$ is given by
\begin{equation}\label{cdf}
_c{\Delta _F}=\max \big\{ \{{}_c\delta_F(1,b) : b \in \gf_{{q}} \} \cup \{\gcd(d,q-1)\} \big\}.
\end{equation}
Moreover, as a generalization of the differential spectrum, the $ c $-differential spectrum of a power function is defined as follows.

\begin{definition}(\cite{WZ}) Let $ F\left( x \right) = {x^d} $ be a power function over $\gf_{q} $ with $c$-differential uniformity $_c\Delta_F$. Denote by $ {}_c{\omega _i} = |\{ {b \in \gf_{q} : {{}_c{\delta _F}( 1,b) = i} } \}|$ for each $ 0 \le i \le {}_c{\Delta _F} $. The $ c $-differential spectrum of $ F $ is defined to be the multi-set \[\mathbb{S} = \{ {{}_c{\omega _i}:{0 \le i \le {}_c{\Delta _F}} ~\mathrm{and}{}~_c{\omega _i} > 0} \}.\]
\end{definition}

The $c$-differential spectrum of $F$ has the following properties.

\begin{theorem}\label{c-identicalequation}
	(\cite{YZ}, Theorem 4)
	Let $ F\left( x \right) = {x^d} $ be a power function over $ \gf_{p^{n}} $ with $c$-differential uniformity $_c{\Delta _F}$ for some $1\neq c\in\gf_{p^{n}}$. Recall that $ {{}_c\omega _i} = |\left\{ {b \in \gf_{p^{n}} : {{}_c{{ \delta}_F}\left(1, b \right) = i}} \right\}| $ for each $ 0 \le i \le {}_c{\Delta _F} $, where $_c\delta_F(1,b)=|\{x\in\gf_{p^{n}}: (x+1)^d-cx^d=b\}|$. We have 
	\begin{equation}\label{c-omegaiomega}
		\sum_{i=0}^{{}_c{\Delta _F}}{}_c\omega_i=\sum_{i=0}^{{}_c{\Delta _F}}i\cdot{}_c\omega_i=p^{n}.
	\end{equation} 	
	Moreover, we have 
	\begin{equation}\label{c-i^2omega}
		\sum\limits_{i = 0}^{{}_c{\Delta _F}} {{i^2}\cdot {}_c{\omega _i}}  = \frac{{{\rm{ }}{{}_cN_4} - 1}}{{p^{n}} - 1}-\gcd(d,p^{n}-1) ,
	\end{equation}
	where
	\begin{equation*}
		{}_cN_{4}=\Bigg|\left\{ {\left( {{x_1},{x_2},{x_3},{x_4}} \right) \in (\gf_{p^{n}})^4 : {\Bigg\{ \begin{array}{ll}
					{x_1} - {x_2} + {x_3} - {x_4} &= 0\\
					x_1^d - cx_2^d + cx_3^d - x_4^d &= 0
		\end{array} } } \right\}\Bigg|. 
	\end{equation*} 
\end{theorem}

It is generally difficult to determine the $c$-differential spectrum of power functions. We know that only a few classes of power functions over odd characteristic finite fields have a nontrivial $c$-differential spectrum. We summarize them in Table \ref{t2}. 

\begin{table}[t]
	\renewcommand{\arraystretch}{1.7}
	\caption{Power functions $F(x)=x^{d}$ over $\gf_{p^{n}}$ with known $c$-differential spectrum}
	\label{t2}
	\centering
	\begin{tabular}{|c|c|c|c|c|}
		\hline
		$p$&$d$ & Conditions & $_{c}\Delta_{F}$ & Ref. \\ 
		\hline\hline
		2&$2^{3m}+2^{2m}+2^{m}-1$&$n=4m$, $0,1\ne c\in \gf_{2^{n}}$, $c^{1+2^{2m}}=1$& $2$&\cite{TZJT}\\\hline
		odd& $p^{k}+1$&$1\ne c\in \gf_{p^{\gcd(n,k)}}$ and $ \frac{n}{\gcd(n,k) }$ is odd, or $ c\notin \gf_{p^{\gcd(n,k)}} $, $ n $ is even and $ k=\frac{n}{2} $ & $2$ &\cite{WZ}\\\hline
		2& $2^{n}-2$& $c\ne 0$, $\tr_{n}(c)=\tr_{n}(c^{-1})=1$ &$2$&\cite{YZ}, Thm 7\\\hline
		2&$2^n-2$ &$c\ne 0$, $\tr_{n}(c)=0$ or $\tr_{n}(c^{-1})=0$ & $3$ &\cite{YZ}, Thm 7\\\hline
		odd&$p^{n}-2$ &any $c$ & 2 or 3 &\cite{YZ}, Thm 8\\\hline
		3&$3^n-3$ & $c=-1$, any $n$& 4 or 6  &\cite{YZ}, Thm 10\\\hline
		odd&$\frac{p^{k}+1}{2}$ & $c=-1$, $ \gcd(n,k)=1$, $ \frac{2n}{\gcd(2n,k) }$ is even & $\frac{p+1}{2}$ & \cite{YZ}, Thms 11-14\\\hline
		5&$\frac{5^{n}-3}{2}$ & $c=-1$, $n\geq 2$& $2$& \cite{YZ}, Thm 18\\\hline
		$p$ & $\frac{p^n-3}2$ & $c=-1$,\ $p^n\equiv3\pmod4$ and $p^{n}>3$ & 2 or 4 & This paper \\\hline
	\end{tabular}
\end{table}

Let $F(x)=x^d$ be a power function over $\gf_{p^{n}}$, where $p$ is an odd prime and $d=\frac{p^{n}-3}{2}$ is a positive integer. The $(-1)$-differential uniformity of $F$ was discussed in \cite{MRSYZ}. It is proved that $_{-1}\Delta_{F}\leq4$. In the rest of this section, we study the $c$-differential uniformity and the $(-1)$-differential spectrum of $F$.

\subsection{The $c$-differential uniformity of $x^{\frac{p^n-3}{2}}$ over $\gf_{p^n}$}\label{a}

We give the following result on the $c$-differential uniformity of $F$ for a general $c\in\gf_{p^n}\setminus\{\pm 1\}$.
	\begin{theorem}\label{c-differential uniform} Let $F(x)=x^d$ be a power function over $\gf_{p^{n}}$, where $p$ is an odd prime and $d=\frac{p^{n}-3}{2}$ is a positive integer. For $\pm 1\neq c \in \gf_{p^{n}}$, we have $_c\Delta_{F}\leq 9$. Moreover, if $p^{n}\equiv 3\pmod4$, then $_c\Delta_{F}\leq 5$.
\end{theorem}

\begin{proof} For a fixed $c\neq \pm1$, to determine the $c$-differential uniformity of $F$, we consider the number of the solutions of
	
	\begin{equation}
		\label{cequ1}
		 _{c}\Delta(x):=(x+1)^{d}-cx^{d}=b
	\end{equation}
	for any $b\in \gf_{p^{n}}$. Denote by $ {\delta_{c}}(b) = |\left\{ {x \in {\gf_{{p^{n}}}}:~_{c}\Delta(x)= b} \right\}|$. In what follows, we will show that $\delta_c(b)\leq 9$ for any $b\in\gf_{{p^{n}}}$. For $b=0$, (\ref{cequ1}) becomes $(1+\frac{1}{x})^d=c$ since $x\neq 0$. Then $\delta_c(0)\leq \mathrm{gcd}(d,p^n-1)=2$.
 Furthermore, for any $b\in\gf_{p^n}^*$, it is obvious that $_c\Delta(0)=1$ and $_c\Delta(-1)=(-1)^{d+1}c$. Then $x=0$ and $x=-1$ cannot be solutions of  \eqref{cequ1} simultaneously since $c\neq \pm 1$. Now we assume that  $x\neq 0, -1$, then \eqref{cequ1} becomes
	\begin{equation}\label{cequ2}
		\chi(x+1)(x+1)^{-1}-c\chi(x)x^{-1}=b,
	\end{equation}
	where $\chi$ denotes the quadratic multiplicative character on $\gf_{p^n}$. Depending on the values of $ \chi(x)$ and  $\chi(x+1)$, we obtain four disjoint cases. In each case, (\ref{cequ2}) can be rewritten as a quadratic equation. We summarize them in Table \ref{t3}.
	
	\begin{table}[t]
		\renewcommand{\arraystretch}{1.7}
		\caption{Simplification of \eqref{cequ2} in four disjoint cases}
		\label{t3}
		\centering
		\begin{tabular}{|c|c|c|c|c|}
			\hline
			Case &   $(\chi(x+1),\chi(x))$ & Corresponding quadratic equations & $ x_{1}x_{2} $ & $ (x_{1}+1)(x_{2}+1) $  \\
	     	\hline\hline
			$ \rm\uppercase\expandafter{\romannumeral1} $ &  $(1,1)$ & $ bx^{2}+(b+c-1)x+c=0 $ & $ \frac{c}{b} $ & $ \frac{1}{b} $  \\\hline
			$ \rm\uppercase\expandafter{\romannumeral2} $ &  $(1,-1)$ & $ bx^{2}+(b-c-1)x-c=0 $ & $ -\frac{c}{b} $ & $ \frac{1}{b} $  \\\hline
			$ \rm\uppercase\expandafter{\romannumeral3} $ &  $(-1,1)$  & $  bx^{2}+(b+c+1)x+c=0 $ & $ \frac{c}{b} $ & $ -\frac{1}{b} $  \\\hline
			$ \rm\uppercase\expandafter{\romannumeral4} $ & $(-1,-1)$ & $ bx^{2}+(b-c+1)x-c=0 $ & $ -\frac{c}{b} $ & $ -\frac{1}{b} $  \\\hline
		\end{tabular}
	\end{table}

	As we see from Table \ref{t3}, \eqref{cequ2} has at most two solutions in each case. Then $\delta_{c}(b) \leq 9$ for any $b\in\gf_{p^n}^*$. By \eqref{cdf} and $\gcd(d, p^{n}-1)\leq 2$, we obtain $_c\Delta_F\leq 9$.
	
	Now we assume $p^{n}\equiv 3\pmod4$, i.e., $\chi(-1)=-1$. It was proved that $\delta_c(0)\leq2$. In the following, we assume that $b\neq 0$. It is mentioned that $_c\Delta(0)=1$ and $_c\Delta(-1)=-c$. For fixed $b\neq 1, -c$, we will show that $\delta_c(b)\leq 5$. If \eqref{cequ2} has at most one solution in each case, then $\delta_c(b)\leq 5$ is obvious. We consider that \eqref{cequ2} has two solutions in one of the four cases. Without loss of generality, if \eqref{cequ2} has two solutions in Case I, namely $x_1$ and $x_2$, then $\chi(x_1x_2)=\chi(\frac{c}{b})=1$ and $\chi((x_1+1)(x_2+1))=\chi(\frac{1}{b})=1$. Hence $\chi(-\frac{c}{b})=-1$, which implies that \eqref{cequ2} has at most one solution in Case II. Otherwise, if $x_3$ and $x_4$ are distinct solutions of \eqref{cequ2} in Case II, then $\chi(x_3)=\chi(x_4)=-1$ and $\chi(x_3x_4)=\chi(-\frac{c}{b})=-1$, which is a contradiction. Similarly, \eqref{cequ2} has at most one solution in Case III (respectively, Case IV) since $\chi(-\frac{1}{b})=-1$. Then $\delta_c(b)\leq 5$ when $b\neq 0, 1, -c$.
	
	Next we consider $b=1$ and $b=-c$. For $b=1$, if $\chi(c)=1$, then we can obtain that \eqref{cequ2} has at most one solution in Case II (respectively, Case IV) since $\chi(-\frac{c}{b})=-1$. Moreover, in Case I, the quadratic equation can be rewritten as $x^{2}=-c(x+1)$, which has no solutions in Case I. Then $\delta_{c}(1) \leq 5$. If $\chi(c)=-1$, we can prove that \eqref{cequ2} has at most one solution in Case I (respectively Case III) and has no solution in Case II. We conclude that Then $\delta_{c}(1) \leq 5$. The proof of $\delta_c(-c)\leq 5$ is similar and omitted. By the above discussions, when $p^n\equiv 3\pmod4$, we obtain that $\delta_c(b)\leq 5$ for any $b\in\gf_{p^n}$, i.e., $_{c}\Delta_F\leq 5$, which completes the proof.
\end{proof}
\begin{remark}

 Let $\alpha$ be a primitive element of $\gf_{{p^{n}}}$. By the numerical results from computer experiments, when $p=7$, $n=4$ and $c=\alpha^{82}$, the $c$-differential uniformity of $F$ is 9. When $p=7,\ n=3$ and $c=\alpha^{327}$, the $c$-differential uniformity of $F$ is 5. Therefore, the bound shown in Theorem \ref{c-differential uniform} is tight.
\end{remark}

\subsection{The $(-1)$-differential spectrum of $x^{\frac{p^n-3}{2}}$}\label{b}
In this subsection, we focus on the $c=-1$. In \cite{MRSYZ}, Mesnager, Riera, St{\u{a}}nic{\u{a}}, Yan, and Zhou proved that the $(-1)$-differential uniformity of $F$ is less than or equal to $4$. In this subsection, we determine the $(-1)$-differential spectrum of $F$ when $p^n\equiv 3 (\mathrm{mod}~ 4)$. For any $b\in\gf_{p^n}$, we consider the $(-1)$-differential equation
\begin{equation}\label{-1diff}
	_{-1}\Delta(x)=(x+1)^d+x^d=b.
\end{equation}
It is not difficult  to see that $x$ is a solution of (\ref{-1diff}) if and only if $-x-1$ is a solution of (\ref{-1diff}). We assert that $\delta_{-1}(b)$ is an even number except
	\begin{equation*}
	b=_{-1}\Delta(-\frac{1}{2})=\left\{\begin{array}{ll}
		4, &~~\mathrm{if}~ \chi(2)=1,\\ 
		-4, &~~\mathrm{if}~ \chi(2)=-1.
	\end{array}\right.
\end{equation*}
In the following, we investigate the values of $\delta_{-1}(_{-1}\Delta(-\frac{1}{2}))$.
\begin{lemma} We have,
	
	(i) if $\chi(2)=1$, then $\delta_{-1}(4)=1$ when $\chi(5)=-1$, and $\delta_{-1}(4)=3$ when $\chi(5)=1$,
	
	(ii) if $\chi(2)=-1$, then $\delta_{-1}(-4)=1$ when $\chi(5)=-1$, and $\delta_{-1}(-4)=3$ when $\chi(5)=1$.
\end{lemma}

\begin{proof}
  
We only prove the first statement. The proof of the second one is similar so we omit it. Now we assume that $\chi(2)=1$. It is mentioned that $p\neq 3$ since $\chi(2)=1$.
Consider
\begin{equation}\label{delta4eqn}
	(x+1)^d+x^d=4.
\end{equation}
It is easy to see that $x=0$ and $x=-1$ are not solutions of (\ref{delta4eqn}). For $x\neq0,-1$, $x^d=\chi(x)x^{-1}$. 
We distinguish the following four cases.
	
	\begin{enumerate}
	\item Case 1. $(\chi(x+1),\chi(x))=(1,1)$. (\ref{delta4eqn}) becomes $x^2+\frac{1}{2}x-\frac{1}{4}=0$, the discriminant of this quadratic equation is $\frac{5}{4}$. If $\chi(5)=-1$, then (\ref{delta4eqn}) has no solution in this case. If $\chi(5)=1$, the quadratic equation $x^2+\frac{1}{2}x-\frac{1}{4}=0$ has two distinct solutions, namely, $x_1$ and $x_2$. Note that $x_1$ and $x_2$ satisfy $\chi(x_1x_2)=\chi(-\frac{1}{4})=-1$, without loss of generality, we assume $\chi(x_1)=1$ and $\chi(x_2)=-1$. Moreover, $x_1+1=2(x_1+\frac{1}{2})^2$, then $\chi(x_1+1)=1$ since $\chi(2)=1$. Therefore, (\ref{delta4eqn}) has one solution in this case if $\chi(5)=1$.
	
	\item Case 2. $(\chi(x+1),\chi(x))=(1,-1)$. (\ref{delta4eqn}) becomes $x^2+x+\frac{1}{4}=0$, one can easily obtain (\ref{delta4eqn}) has one solution $x=-\frac{1}{2}$ in this case.
	
	\item Case 3. $(\chi(x+1),\chi(x))=(-1,1)$. (\ref{delta4eqn}) becomes $x^2+x-\frac{1}{4}=0$, then $x(x+1)=\frac{1}{4}$, which contradicts to $\chi(x(x+1))=-1$. Hence (\ref{delta4eqn}) has no solution in this case.
	
	\item Case 4. $(\chi(x+1),\chi(x))=(-1,-1)$. Note that $x$ is a solution of (\ref{delta4eqn}) in this case if and only if $-x_0-1$ is a solution in Case 1. Then 
	(\ref{delta4eqn}) has no solution in this case if $\chi(5)=-1$ and one solution in this case if $\chi(5)=1$.
	\end{enumerate}

	By discussions as above, the desired result follows.	
\end{proof}

One can immediately deduce the following corollary.
\begin{corollary}\label{-1omega}
	With the notation as above, we have $_{-1}\omega_{1}=1$ and $_{-1}\omega_{3}=0$ when $\chi(5)=-1$, and $_{-1}\omega_{1}=0$, $_{-1}\omega_{3}=1$ when $\chi(5)=1$.
\end{corollary}

Denote by $_{-1}N_4$ the number of solutions $(x_1,x_2,x_3,x_4)\in(\gf_{p^n})^4$ of the following equation system 
\begin{equation}\label{c-equationsystem}
	\left\{\begin{array}{ll}
	{x_1} - {x_2} + {x_3} - {x_4} &= 0 \\ 
	x_1^d + x_2^d - x_3^d - x_4^d &= 0,
	\end{array}\right.
\end{equation} 
where $d=\frac{p^n-3}{2}$ is a positive integer. To investigate the $(-1)$-differential spectrum of $F$, we should determine $_{-1}N_4$. When $d$ is even, the equation system (\ref{c-equationsystem}) can be rewritten as
\begin{equation*}
	\left\{\begin{array}{ll}
		{x_1} - (-x_3)~+ (-x_2)~- {x_4} &= 0 \\ 
		x_1^d - (-x_3)^d + (-x_2)^d - x_4^d &= 0.
	\end{array}\right.
\end{equation*}
Then by Theorem \ref{N_4value}, \[_{-1}N_{4}=N_4= \left\{ \begin{array}{ll}
	\frac{1}{2}(5p^{2n}+(\lambda_{p,n}-10)p^n-(\lambda_{p,n}-7)),&~\mathrm{if}~\chi(5)=-1,\\
	\frac{1}{2}(5p^{2n}+(\lambda_{p,n}+6)p^n-(\lambda_{p,n}+9)),&~\mathrm{if}~\chi(5)=1,
\end{array} \right.\]
where $p^n\equiv3\pmod4$ and $p^n>3$. The $(-1)$-differential spectrum of $F$ is determined as follows.
\begin{theorem}\label{c-differential spectrum}
	Let $F(x)=x^{\frac{p^{n}-3}{2}}$ be the power function over $\gf_{p^{n}}$, where $p^n\equiv3\pmod4$ and $p^n>3$. The $(-1)$-differential spectrum of $F$ is
	\begin{align*}
		\s=\{&_{-1}\omega_{0}=\frac{1}{16}(9p^n+\lambda_{p,n}-15),\\&_{-1}\omega_{1}=1,\\&
		_{-1}\omega_{2}=\frac{1}{8}(3p^n-\lambda_{p,n}+3),\\&
		_{-1}\omega_{4}=\frac{1}{16}(p^n+\lambda_{p,n}-7)\}
	\end{align*}
	when $\chi(5)=-1$, and is
	\begin{align*}
		\s=\{&_{-1}\omega_{0}=\frac{1}{16}(9p^n+\lambda_{p,n}+9),\\&_{-1}\omega_{2}=\frac{1}{8}(3p^n-\lambda_{p,n}-13),\\&
		_{-1}\omega_{3}=1,\\& _{-1}\omega_{4}=\frac{1}{16}(p^n+\lambda_{p,n}+1)\}
	\end{align*}
	when $\chi(5)=1$, where $\lambda_{p,n}$ was defined in \eqref{sum}. When $p^n\equiv3\pmod4$ and $p^n\neq 7, 27$, the $(-1)$-differential uniformity of $F$ is  $4$. When $p^n=7$ and $p^n=27$, $F$ is AP$_{-1}$N.
\end{theorem}
\begin{proof}
	By \eqref{c-i^2omega} and the value of $_{-1}N_{4}$, we have
	\begin{equation}\label{c-i^2omegaidenty}
		\sum_{i=0}^{_{-1}\Delta_{F}}i^{2}\cdot_{-1}\omega_{i}=\left\{\begin{array}{ll}
		\frac{1}{2}(5p^n+\lambda_{p,n}-9), & ~\mathrm{if}~ \chi(5)=-1 \\ 
		\frac{1}{2}(5p^n+\lambda_{p,n}+7), & ~\mathrm{if}~ \chi(5)=1 .
		\end{array}\right.
	\end{equation}
	since $\gcd(d,p^n-1)=2$, where $\lambda_{p,n}$ was defined in \eqref{sum}. By Corollary \ref{-1omega}, \eqref{c-omegaiomega} and \eqref{c-i^2omegaidenty}, $_{-1}\omega_{0}, _{-1}\omega_{2}$ and $_{-1}\omega_{4}$ satisfy  
	\begin{equation*}
		\left\{\begin{array}{ll}
			_{-1}\omega_{0}+_{-1}\omega_{2}+_{-1}\omega_{4}&=p^n-1\\ 
			2_{-1}\omega_{2}+4_{-1}\omega_{4}&=p^{n}-1\\
			4_{-1}\omega_{2}+16_{-1}\omega_{4}&=\frac{1}{2}(5p^n+\lambda_{p,n}-11)
		\end{array}\right.
	\end{equation*}
	when $\chi(5)=-1$, and they satisfy
	\begin{equation*}
		\left\{\begin{array}{ll}
			_{-1}\omega_{0}+_{-1}\omega_{2}+_{-1}\omega_{4}&=p^n-1\\ 
			2_{-1}\omega_{2}+4_{-1}\omega_{4}&=p^{n}-3\\
			4_{-1}\omega_{2}+16_{-1}\omega_{4}&=\frac{1}{2}(5p^n+\lambda_{p,n}-11)
		\end{array}\right.
	\end{equation*}
	when $\chi(5)=1$. By solving the above two equation systems, the $(-1)$-differential spectrum of $F$ can be determined. Moreover, by Theorem \ref{bound}, we obtain $|\lambda_{p,n}|\leq 4p^{\frac{n}{2}}$ and then
	
	\[_{-1}\omega_{4}\geq \frac{1}{16}(p^n+\lambda_{p,n}-7)\geq \frac{1}{16}(p^n-4p^{\frac{n}{2}}-7)>0\]
	when $p^n\geq 29$.
	The numerical results show that $_{-1}\omega_4\geq 1$ for $p^n\in\{11,19,23\}$. We assert that $_{-1}\omega_4\geq 1$ for all $p^n \equiv 3 (\mathrm{mod}~4)$ and $p^n\neq 3,7,27$. 
	This implies that the $(-1)$-differential uniformity of $F$ is $4$. For $p^n=7$ and $p^n=27$, it can be calculated that $\lambda_{7,1}=0$ and $\lambda_{3,3}=-20$. By the $(-1)$-differential spectrum of $F$, the $(-1)$-differential uniformity of $F$ follows, which completes  the proof.
\end{proof}


Below, we present some examples.
\begin{example}
\begin{itemize}
	\item For $p=3$ and  $n=5$, the power function $F(x)=x^{120}$ over $\gf_{3^{5}}$ is differentially $(-1,4)$-uniform with $(-1)$-differential spectrum
	\begin{equation*}
		\s=\{_{-1}\omega_{0}=136, _{-1}\omega_{1}=1, _{-1}\omega_{2}=91, _{-1}\omega_{4}=15\}.
	\end{equation*}

	\item  For $p=7$ and  $ n=3$, the power function $F(x)=x^{170}$ over $\gf_{7^{3}}$ is differentially $(-1,4)$-uniform with $(-1)$-differential spectrum
	\begin{equation*}
		\s=\{_{-1}\omega_{0}=192,_{-1}\omega_{1}=1,_{-1}\omega_{2}=129, _{-1}\omega_{4}=21\}.
	\end{equation*}

	\item For $p=11$ and  $ n=3$, the power function $F(x)=x^{664}$ over $\gf_{11^{3}}$ is differentially $(-1,4)$-uniform with $(-1)$-differential spectrum
	\begin{equation*}
		\s=\{_{-1}\omega_{0}=745,_{-1}\omega_{2}=506,_{-1}\omega_{3}=1,_{-1}\omega_{4}=79\}.
	\end{equation*}
	\end{itemize}
\end{example}

\section{Concluding remarks}\label{E}
In this paper, we investigated the differential properties of the power function $x^{\frac{p^n-3}{2}}$ over $\gf_{p^n}$ with odd $p$. With such specific settings, the studied power functions are APN functions, offering the lowest differential uniformity and the nontrivial differential spectrum.  Firstly, we determined the differential spectrum of $F$ when $p^n\equiv3\pmod4$, $p^n>7$ and $p^n\ne27$ in terms of quadratic character sum $\lambda_{p,n}$. Secondly, we  used similar strategy to provide the $(-1)$-differential spectrum $F$ when $p^n\equiv3\pmod4$ and $p^n>3$. We highlight that our approach to determining the $c$-differential uniformity of $F$ can be used to evaluate the $c$-differential spectrum of other power functions. Deriving low $c$-differentially uniformity functions could be attractive theoretically. Finally, we emphasize that our study on the considered concept, viewed as an extension of the classical concept of the differential spectrum, is based on elliptic curves and Galois theories over finite fields. It contributes to theoretically completing and enriching the literature by evaluating sums of specific characters closely related to elliptic curves and determining the number of solutions of specific equations systems.
 


\end{document}